\documentclass[shortpaper, 9pt,journal,twoside,web]{ieeecolor}  

\usepackage{amsfonts}       
\usepackage{nicefrac}       
\usepackage{microtype}      
\usepackage{bm}
\usepackage{float}
\usepackage[ruled,vlined]{algorithm2e}
\usepackage{mathtools}
\usepackage{generic}

\usepackage{amsthm}
\usepackage{amssymb}
\usepackage{amsmath}
\usepackage{breqn}
\usepackage{caption}
\usepackage{wrapfig}
\usepackage{tikz}\usetikzlibrary{arrows}
\usepackage{bbm}
\usepackage[bb=boondox]{mathalfa}
\newtheorem{assumption}{Assumption}
\theoremstyle{plain}
\newtheorem{theorem}{Theorem}[section]
\theoremstyle{plain}

\theoremstyle{plain}
\theoremstyle{definition}

\theoremstyle{plain}
\theoremstyle{plain}

\theoremstyle{plain}
\newtheorem{remark}[theorem]{Remark}
\theoremstyle{definition}

\newcounter{stepskim}
\newcommand{\stepkim}[1][\empty]
{\stepcounter{stepskim}%
 \par
 \hangindent=4em
 \hangafter=1
 \makebox[4em][l]{\textit{Step \arabic{stepskim}:}}%
 \ifx#1\empty\else #1 --\fi
}
{\setcounter{stepskim}{0}%
 \let\item=\step
 \parskip=0.25\baselineskip
 \parindent=0pt}%
{\par}

\newcommand{\RR}{\mathbb{R}}

\newcommand{\mDefFunction}[3]{#1: #2 \rightarrow #3}

\newcommand{\mDef}{\coloneqq}
\newcommand{\mDistribution}[1]{\mathcal Q_{#1}}

\newcommand{\mExpectation}[1]{\mathbb E\left[#1\right]}
\newcommand{\mExpectationSmall}[1]{\mathbb E [#1]}
\newcommand{\mExpectationExp}[2]{\mathbb{E}_{#1}\left[#2\right]}
\newcommand{\mExpectationExpSmall}[2]{\mathbb{E}_{#1}[#2]}

\newcommand{\mData}{\mathbb D}

\newcommand{\mOnes}[1]{\mathbb 1_{#1}}

\newcommand{\mKdim}{\mathrm{dim}_K}
\newcommand{\mEdim}{\mathrm{dim}_E}

\newcommand{\mPa}{\theta}
\newcommand{\mV}[3]{V^{#1}_{{#2},{#3}}}

\newcommand{\mR}{R}
\newcommand{\mPo}{\mUpredOpt}
\newcommand{\mRe}[1]{{\Delta}_{#1}}
\newcommand{\mtRe}[1]{\tilde {\Delta}_{#1}}
\newcommand{\mCre}{\mathrm{CR}}

\newcommand{\mMpcCost}[2]{J^#1_#2}

\newcommand{\mUpred}{\bar u}
\newcommand{\mUpredOpt}{\bar u^*}

\newcommand{\mXpred}{\bar x}

\newcommand{\XX}{\mathbb X}
\newcommand{\UU}{\mathbb U}

\newcommand{\NN}{\mathcal {N}}

\newcommand{\OO}{\mathcal{O}}
\newcommand{\BB}{\mathcal{B}}

\renewcommand{\BB}{\mathcal{B}}

\newcommand*{\END}{\hfill\ensuremath{\lhd}}

\newcommand{\re}{\color{black}}

\tikzstyle{block} = [draw, rectangle,
minimum height=1cm, minimum width=3cm, text width=2.5cm, align=center]
\tikzstyle{sum} = [draw, fill=blue!20, circle, node distance=1cm]
\tikzstyle{input} = [coordinate]
\tikzstyle{output} = [coordinate]
\tikzstyle{pinstyle} = [pin edge={to-,thin,black}]
\makeatletter
\newif\ifmygrid@coordinates
\tikzset{/mygrid/step line/.style={line width=0.80pt,draw=gray!80},
	/mygrid/steplet line/.style={line width=0.25pt,draw=gray!80}}
\pgfkeys{/mygrid/.cd,
	step/.store in=\mygrid@step,
	steplet/.store in=\mygrid@steplet,
	coordinates/.is if=mygrid@coordinates}
\def\mygrid@def@coordinates(#1,#2)(#3,#4){%
	\def\mygrid@xlo{#1}%
	\def\mygrid@xhi{#3}%
	\def\mygrid@ylo{#2}%
	\def\mygrid@yhi{#4}%
}
\newcommand\DrawGrid[3][]{%
	\pgfkeys{/mygrid/.cd,coordinates=true,step=1,steplet=0.2,#1}%
	\draw[/mygrid/steplet line] #2 grid[step=\mygrid@steplet] #3;
	\draw[/mygrid/step line] #2 grid[step=\mygrid@step] #3;
	\mygrid@def@coordinates#2#3%
	\ifmygrid@coordinates%
		\draw[/mygrid/step line]
		\foreach \xpos in {\mygrid@xlo,...,\mygrid@xhi} {%
				(\xpos,\mygrid@ylo) -- ++(0,-3pt)
				node[anchor=north] {$\xpos$}
			}
		\foreach \ypos in {\mygrid@ylo,...,\mygrid@yhi} {%
				(\mygrid@xlo,\ypos) -- ++(-3pt,0)
				node[anchor=east] {$\ypos$}
			};
	\fi%
}
\makeatother
\newtheorem{lemma}{Lemma}

\begin{document}
\title{Cautious Bayesian MPC: Regret Analysis and Bounds on the Number of Unsafe Learning Episodes}
\author{Kim P.~Wabersich and Melanie~N.~Zeilinger
\thanks{This work was supported by the Swiss National Science Foundation under grant no. PP00P2 157601/1.}
\thanks{K. P. Wabersich and M. N. Zeilinger are with the Institute for Dynamical
Systems and Control ETH Zurich, Zurich,
Switzerland (e-mail: \{wkim$\vert$mzeilinger\}@ethz.ch). }}

\maketitle

\begin{abstract}
  This paper investigates the combination of model predictive control (MPC) concepts and posterior sampling techniques and proposes a simple constraint tightening technique to introduce cautiousness during explorative learning episodes. The provided theoretical analysis in terms of cumulative regret focuses on previously stated sufficient conditions of the resulting `Cautious Bayesian  MPC' algorithm and shows Lipschitz continuity of the future reward function in the case of linear MPC problems. In the case of nonlinear MPC problems, it is shown that commonly required assumptions for nonlinear MPC optimization techniques provide sufficient criteria for model-based RL using posterior sampling. Furthermore, it is shown that the proposed constraint tightening implies a bound on the expected number of unsafe learning episodes in the linear and nonlinear case using a soft-constrained MPC formulation. The efficiency of the method is illustrated using numerical examples.
\end{abstract}

\section{INTRODUCTION}

Driven by a constantly increasing research and development effort in the field of autonomous systems, including, e.g., autonomous driving, service robotics, or production processes in chemical or biological industry branches, the number of constrained control problems is growing steadily, which motivates research efforts towards automated and efficient synthesis procedures of high-performance control algorithms.
While significant progress in this context has been made in the area of reinforcement learning (RL) comparably few methods address systems with continuous state and input spaces and can take into account system constraints. In addition they often require up to hundreds of thousands of rollouts, see, e.g., \cite{Wei2021,achiam2017constrained}, prohibiting their direct application to most real-world control engineering applications.
In addition to related learning problems for general nonlinear systems, there exist a variety of results for the linear quadratic regulator (LQR) problem, see, e.g.,  \cite{Ouyang2020} for the case of known stage cost and infinite horizon or \cite[Appendix D]{Osband2014} for finite-time horizons. While LQR-based techniques can provide an efficient learning procedure, a principled incorporation of state and input constraints can involve a significant amount of manual cost shaping. 
Learning-based control design for constrained systems has also been approached from a control engineering perspective. Particularly in the case of safety-critical control problems, model predictive control (MPC) techniques~\cite{borrelli2003constrained} have shown significant impact in both, industrial and research-driven applications. Due to its principled controller synthesis procedures and professional software tools~\cite[Chapter 8]{Rawlings2017}, MPC offers an important framework for a variety of learning-based control tasks~\cite{hewing2020}, including larger-scale systems. However, it should be noted that challenging control problems can be difficult to formulate in the form of a learning-based MPC policy, e.g., in the case of sparse or discontinuous dynamics and rewards.
Since MPC heavily relies on a sufficiently accurate prediction and reward model of the system, research in learning-based MPC mainly focuses on automatically improving the model quality. To this end, common methods either rely on available system data or include an exploration mechanism similar to RL, see, e.g.~\cite[Section 3]{hewing2020}, for an overview.
In addition, recent learning strategies include the optimization of model and cost function parameters of an MPC using automatic differentiation~\cite{NIPS2018_8050}, sensitivity analysis~\cite{Gros2019}, or Bayesian optimization \cite{Neumann-Brosig2019}. While the underlying concepts are promising, their theoretical properties still need to be investigated.
%
 A concept for combining the merits of posterior sampling based RL~\cite{Osband2014} and learning-based MPC~\cite{hewing2020} has been introduced in a Bayesian MPC scheme~\cite{Wabersich2020}, enabling practical and scalable RL for industrial applications with state and input constraints. However, the fundamental regularity assumption~\cite[Section 6.1]{Osband2014} required to obtain rigorous regret bounds was not established. Compared with constrained RL techniques~\cite{Liu2021,Efroni2020}, the Bayesian MPC approach in \cite{Wabersich2020} is further missing bounds on the number of unsafe learning episodes.
\paragraph*{Contributions}

  In order to complete the Bayesian MPC theory and justify its central assumption~\cite[Assumption 1]{Wabersich2020} to obtain cumulative regret bounds, Section~\ref{subsec:regularity_linear_concave} provides a theoretical analysis based on multiparametric optimization theory in the case of linear system dynamics, convex constraints, and linear or concave quadratic rewards. Then, Section~\ref{subsec:regularity_nonlinear} extends the analysis to nonlinear MPC by connecting~\cite[Assumption~1]{Wabersich2020} with common assumptions from nonlinear MPC and shows that these assumptions are also sufficient for nonlinear Bayesian MPC.
  The presented cautious Bayesian MPC formulation uses a simple state constraint tightening as introduced in Section~\ref{subsec:soft_mpc} that allows to rigorously relate the expected number of unsafe learning episodes to the cumulative performance regret bound in Section~\ref{subsec:unsafe_learning_episodes}, recovering similar bounds on the number of unsafe learning episodes as derived, e.g., in the context of RL~\cite{Wei2021,Liu2021,Efroni2020}. Compared to the original formulation in \cite{Wabersich2020}, this guarantees fewer constraint violations as the learning process advances while maintaining the fast practical learning convergence due to the MPC policy parametrization.


\subsection{Preliminaries}
  \emph{Notation:} We denote the $i$-th element of a vector $c\in\RR^n$ as $c_i$.
  A vector with $n$ elements equal to $1$ is denoted by $\mOnes{n}$ and if it is clear from the context
  we write $\mOnes{}$. \END

  We consider dynamical systems that can be modeled as
  \begin{align}\label{eq:transition}
    x(t+1) &= f(x(t),u(t);\mPa_{F})+\epsilon_{F}(t),
      \qquad t\in\mathbb N
  \end{align}
  with parameters $\mPa_F$, $\sigma_F$-sub-Gaussian zero
  mean i.i.d. noise $\epsilon_{F}(t)$, and random initial condition
  $x(0)$. The system states
  $x(t)$ are constrained to a desired admissible state constraint
  set of the form $\XX\mDef\{x\in\RR^{n} | g_x(x) \leq \mOnes{}\}$
  and the inputs $u(t)$ are limited to hard constraints of the form
  $\UU\mDef\{u\in\RR^m | g_u(u) \leq \mOnes{}\}$
  where $\mDefFunction{g_x}{\RR^n}{\RR^{n_x}}$ and
  $\mDefFunction{g_u}{\RR^m}{\RR^{n_u}}$.
  The system is equipped with a time-varying reward signal of the form
  \begin{align}
    r(t, x(t), u(t);\mPa_\mR) &= \ell(t, x(t), u(t);\mPa_\mR) + \epsilon_\mR(t)
      \label{eq:reward_model}
  \end{align}
  where $\epsilon_\mR(t)$ is $\sigma_R$-sub-Gaussian zero mean i.i.d. noise
  and $\mPa_\mR$ parameterizes a mean reward function at time $t\in\mathbb N$.
  %
  In case of perfectly known parameters
  $\mPa_{F}$ and $\mPa_\mR$, the goal is to find a control policy
  $\mDefFunction{\pi}{\mathbb N\times\XX}{\UU}$ such that application of
  $u(t) = \pi(t,x(t))$ maximizes the sum of reward signals
  starting from a given initial condition $x(0)$ over a
  finite horizon of $T$ time steps:
  \begin{align}\label{eq:nominal_objective}
    \max_\pi \quad \mExpectationExp{E}{\sum_{t=0}^{T-1} \ell(t,x(t),u(t);\mPa_\mR)}
  \end{align}
  subject to~\eqref{eq:transition} with
  $E \mDef [\epsilon_{F}(0), .., \epsilon_{F}(T-2)]$.
  Importantly, maximization of~\eqref{eq:nominal_objective} needs to be performed
  while taking into account state and input constraints, i.e.
  $x(t)\in\XX$ and $u(t)\in\UU$ for all $t=0,...,T-1$.

  %

\section{CAUTIOUS BAYESIAN MPC}\label{sec:cautious_bmpc}
  \subsection{Cautious soft-constrained model predictive control}\label{subsec:soft_mpc}
  In the idealized case of perfect system parameter knowledge, an
  approximate policy $\pi$ to maximize~\eqref{eq:nominal_objective}
  can be obtained by repeatedly solving a simplified model predictive
  control (MPC) problem, initialized at the currently measured
  state $x(t)$.
  While MPC formulations vary greatly in their complexity, a simple
  formulation as originally proposed by~\cite{Mitter1966} provides
  sufficient practical properties in terms of performance and
  constraint satisfaction for many applications.
  
  Thereby, we optimize over an input sequence $\{\mUpred_{k|t}\}$ such that the corresponding nominal predicted states ${\mXpred_{k|t}}$ neglecting additive disturbances satisfy the system constraints\re. The resulting MPC problem is given by
  \begin{subequations}\label{eq:mpc_problem}
  \begin{align}
    \mMpcCost{\mPa}{t}(x) \mDef 
      \max_{\{\mUpred_{k|t}\}} &~~ \sum_{k=t}^{T-1} \ell(k, \mXpred_{k|t}, \mUpred_{k|t};\mPa_\mR)-I(\rho_{k|t})\label{eq:mpc_problem_reward} \\
      \text{s.t.} &~~ \mXpred_{t|t} = x,  \label{eq:mpc_problem_initial_condition}\\ 
                  &~~ \text{for all }  k = t,..,T-2: \nonumber \\
                  &~~ \quad \mXpred_{k+1|t} = f(\mXpred_{k|t}, \mUpred_{k|t}; \mPa_F),
                    \label{eq:mpc_problem_transition}\\
                  &~~ \text{for all }  k = t,..,T-1: \nonumber \\
                  &~~ \quad\mXpred_{k|t} \in \bar\XX_{k-t}(\rho_{k|t})
                    \text{ with } \rho_{k|t}\geq 0 \label{eq:mpc_problem_state_constraint}, \\
                  &~~ \quad\mUpred_{k|t} \in \UU,\label{eq:mpc_problem_action_constraint}
  \end{align}
  \end{subequations}
  and can be efficiently solved online based on the current system state $x(t)$ using tailored MPC solvers~\cite[Chapter 8]{Rawlings2017}. Ideally, the prediction horizon $T$ equals the task length, yielding a shrinking horizon MPC. For long task horizons $T$, a common approximation in MPC is to select a smaller prediction horizon and to operate in a receding horizon fashion, see, e.g. \cite[Section 2.2]{hewing2020}.
  %
   Different from~\cite{Wabersich2020}, we propose two modifications of the state constraint formulation~\eqref{eq:mpc_problem_state_constraint}\re . First, we use a tightened state constraint common in MPC for uncertain systems to foster  practical closed-loop constraint satisfaction,  see, e.g.~\cite{marruedo2002input}, for the case of bounded uncertainties. By optimizing state trajectories subject to a tightened state constraint set, i.e. $g_x(x) \leq (1-\delta_{k-t})\mOnes{}$ with non-decreasing $0<\delta_{k-t}\leq 1$ along prediction time steps ${k}$ at time $t$, we gain a safety margin to compensate for uncertain model parameters $\mPa_F$ and unknown external disturbances $\epsilon_F$ before state constraint violation occurs at some time step $t$ in the future, i.e. $g_x(x(t))\nleq\mOnes{}$.
  As a second modification, we soften the tightened state constraint in~\eqref{eq:mpc_problem_state_constraint} and include the extra negative reward term $-I(\rho)$ on the constraint relaxation in~\eqref{eq:mpc_problem_reward} to ensure  recursive feasibility of problem~\eqref{eq:mpc_problem} as similarly done in~\cite{Zeilinger2014}. The penalty  ideally realizes a so-called exact penalty function as proposed by~\cite{Kerrigan2000}.
  The resulting cautious soft-constraint formulation is given as  $\mXpred_{k|t}\in\bar\XX_{k-t}(\rho) := \{x\in\RR^n | g_x(x) \leq  (1-\delta_{k-t})\mOnes{} + \rho\}$ with parameters $\delta_{k-t}\in\RR$, $\delta_{k-t}>0$ defining the degree of cautiousness,  slack variable $\rho \in \RR^{n_x}$, $\rho\geq 0$, and  the constraint violation penalty $I(\rho) = c_{1}^\top \rho + c_{2}\rho^\top \rho$ with penalty weights $c_{1}\in\RR^{n_x}$, $c_1>0$, $c_2\in \RR$, $c_2\geq 0$ \cite{Kerrigan2000}\footnote{ The potentially difficult selection of the penalty weights $c_1$, $c_2$ can be avoided using a `separation of objectives' approach, where~\eqref{eq:mpc_problem} is first maximized w.r.t. $-I(\rho_{k|t})$ to obtain a minimal softening $\{\rho^*_{k|t}\}$, which can be fixed to optimize the expected reward in~\eqref{eq:mpc_problem} in a second optimization problem, allowing to select, e.g., $c_1=\mOnes{}$, $c_2=1$.}. Sufficiently large values of $c_1$ yield an exact penalty, which ensures that~\eqref{eq:mpc_problem} recovers the original hard constrained solution in the convex case whenever possible while guaranteeing feasibility of~\eqref{eq:mpc_problem} for any $x\in\RR^n$ if necessary.
  %
  Using an imperfect estimate $\tilde \mPa=(\tilde \mPa_\mR, \tilde \mPa_F)$ of the
  true system parameters $\mPa \mDef (\mPa_\mR,\mPa_F)$ in the MPC problem~\eqref{eq:mpc_problem},
  we denote the expected closed-loop future reward, including a weighted constraint violation penalty,
  at time $t$ and state $x$ as
  \begin{align}\label{eq:reward}
    \mV{\mPa}{\tilde\mPa}{t}(x)\mDef
      \mExpectationExp{E}{
        \sum_{j=t}^{T-1} r(j, x(j), u(j);\mPa_\mR)-I(\rho(j))
          ~\middle|~(\star)}
  \end{align}
  conditioned on $(\star)$ defined as
  \begin{align*}
      x(t) &= x,  \qquad x(j+1) = f(x(j), u(j); \mPa_F) + \epsilon_F(j), \\
      u(j) &= \mPo(j, x(j);\tilde \mPa), \qquad \rho(j) = \min_{\rho\geq 0}{\rho} \text{ s.t. } x(j) \in \bar\XX_0(\rho),
  \end{align*}
  with $E \mDef [\epsilon_{\mR}(t), .., \epsilon_{\mR}(T-1),\epsilon_{F}(t), .., \epsilon_{F}(T-2)]$
  and $\mPo(j, x(j);\tilde\mPa) \mDef \mUpredOpt_{j|j}(x(j);\tilde\mPa)$,
  being the first element of the optimal input sequence of the MPC
  problem~\eqref{eq:mpc_problem} at time step $j$ with parameters $\tilde\mPa$ and $\rho(j)$ the
  required softening of state constraints.
  \par
  
  Similarly as in previous work on model-based RL, e.g., \cite{Osband2014,Wabersich2020}, the value function \eqref{eq:reward} combines a performance measure for safe trajectories, i.e. $I(\rho(j))=0$ for $x(j)\in\bar\XX_0(0)\subseteq\XX$, with a measure for safety violations $x(j)\notin\XX \Rightarrow \rho(j)>0 \Rightarrow I(\rho(j))\gg 0$. More precisely, the tightened constraints $\bar\XX_0(\rho)$ in $(\star)$, which equal~\eqref{eq:mpc_problem_state_constraint} for $k=t$, ensure that $x(j)\notin\XX \Rightarrow \rho(j)>0$, and will be the main mechanism to derive a sublinear bound on the expected number of unsafe learning episodes in Section~\ref{subsec:unsafe_learning_episodes}. As it will be shown in Theorem~\ref{thm:total_number_of_unsafe_episodes}, the linear tightening $\delta_0>0$ can be selected arbitrarily small with the limiting case $\lim_{\delta_0\rightarrow 0}\bar\XX_0(0)=\XX$. The combined value function~\eqref{eq:reward} can therefore reflect the reward without constraint tightening arbitrarily closely. This concept relates to RL for constrained problems as it ensures the existence of a `strictly safe baseline policy'~\cite[Assumption 2.2]{Liu2021} or a `slater point' \cite[Assumption 2]{Efroni2020} if $\mExpectation{\rho(j)}=0$.

\subsection{Reinforcement learning problem}\label{subsec:rl_problem}
  We consider the case of unknown transition and reward distributions that are parametric according
  to \eqref{eq:transition}, \eqref{eq:reward_model}.
  During each learning episode $e=0,1,..,N-1$ we need to provide a control policy that trades-off information extraction and knowledge exploitation when applied to the system~\eqref{eq:transition} at each time step $t=0,1,..,T-1$ starting from $x(0)$.
  We, therefore, assume access to prior information about the system parameterization such as production tolerances, to be given as $(\mPa_F, \mPa_\mR)\sim\mDistribution{\mPa}$.
  Collected data up to $N$ episodes is denoted by
  \begin{align}\label{eq:data}
    \mData_{N}\mDef 
    \left\{
    \left(  
        t,
        x_{t,e},
        u_{t,e},
        x_{t+1,e},
        r_{t,e}
    \right)_{t=0}^{T-1}
    \right\}_{e=0}^{N-1}.
  \end{align}
  Conditioned on collected data~\eqref{eq:data}, the corresponding
  posterior distribution up to episode $e$ is denoted by
  $\mPa_e \sim \mDistribution{\mPa|\mData_e}$.
  Based on the acquired data over $N$ episodes,
  the performance of the RL algorithm is measured in terms of
  the expected Bayesian cumulative regret
  \begin{align}\label{eq:cumulative_regret}
    \mCre(N)\mDef
      \mExpectationExp{\mPa, \mPa_e, \mData_e}{ \sum_{e=0}^{N-1} \mRe{e} }
  \end{align}
  with episodic regret defined as
  \begin{align}\label{eq:instant_regret}
    \mRe{e} \mDef \mExpectationExp{x}{
      \mV{\mPa}{\mPa}{0}(x)
      -
      \mV{\mPa}{\mPa_e}{0}(x)
    }.
  \end{align}
  Using the notation of the expected future reward in~\eqref{eq:reward},
  the cumulative regret \eqref{eq:cumulative_regret} quantifies
  the expected performance deviation between the MPC-based RL algorithm
  using episodically updated model parameters $\mPa_e$ and the
  optimal MPC-based policy with access to the true
  parameters $\mPa$ of the underlying system.

  \subsection{Cautious Bayesian MPC algorithm}\label{subsec:algorithm}
  Following the concept introduced in \cite{Wabersich2020}, we propose to combine model-based RL using posterior sampling investigated by~\cite{Osband2014} with a cautious model predictive control policy parametrization as described in Section~\ref{subsec:soft_mpc} to obtain a new class of model-based RL policies,  called Cautious Bayesian MPC, which allows to bound the number of unsafe learning episodes.
  At the beginning of each learning episode $e$ we sample transition and reward parameters $\mPa_e$ according to their  posterior distribution that results from the prior distribution $\mDistribution{\mPa}$ together with observed data $\mData_e$.
  \begin{algorithm}[t]
    \SetAlgoLined
      \KwData{Parametric model $f$, $\ell$; Prior $\mDistribution{\mPa}$}
      Initialize $\mData_0 = \emptyset$

      \For{episodes $e = 0, 1, .., N-1$}{
        sample $\mPa_e \sim \mDistribution{\mPa|\mData_e}$


        \For{time steps $t = 0, 1, .., T-1$}{
          apply $u(t) = \mPo(t, x(t);\mPa_e)$

          measure objective and state 
        }

        extend data set to obtain $\mData_{e+1}$
      }
    \caption{Bayesian MPC algorithm}
  \end{algorithm}
  The sampled parameters yield an MPC problem parametrization~\eqref{eq:mpc_problem} that would correspond to a system and reward with true parameters $\mPa_e$. However, since $\mPa_e \neq \mPa$, such an MPC policy might be inconsistent with the underlying system to be controlled, in particular if the posterior parameter variance is large.
  In this case, the sampled policy is likely to cause explorative closed-loop behavior, producing information-rich data.  Compared to using, e.g., the current maximum a-posteriori estimate of the parameters as done in most learning-based MPC approaches, which we refer to as nominal posterior MPC \cite{hewing2020}, the algorithm therefore generates explorative behavior in case of large posterior parameter uncertainties.
  As soon as the task-relevant parameter distributions begin to cumulate around the corresponding true process parameters,
  the parameter samples will start to cumulate as well, implying convergence of the sampled MPC performance to the MPC performance with perfectly known parameters. We provide a rigorous analysis of this effect in Sections~\ref{subsec:regularity_linear_concave} and~\ref{subsec:regularity_nonlinear}, which is one of the key ingredients to obtain a regret bound.
  In addition to performance, if the MPC policy using the true system parameters is capable of ensuring cautious constraint satisfaction in expectation, i.e. if $\mExpectationExp{E}{I(\rho(t))}=0$ for $\delta_0 >0$ under $u(j) = \mUpredOpt(j,x(j),\mPa)$ in~\eqref{eq:reward}, then we can additionally bound the expected number of unsafe learning episodes, i.e. the number of episodes in which the state constraints are violated. As formalized in Section~\ref{subsec:unsafe_learning_episodes}, this can be achieved through a sufficiently large $c_1$ in the exact penalty $I(.)$ in \eqref{eq:reward} together with a bound on the instant regret~\eqref{eq:instant_regret}.
  \begin{remark}
    
    Efficient algorithms for solving the MPC problem~\eqref{eq:mpc_problem} already scale to larger-scale industrial problems, see, e.g., \cite{kumar2021industrial}, which considers 256 states and 32 input dimensions. Furthermore, distributed structures as, e.g., considered in the numerical example in Section~\ref{sec:numerical_examples} or in~\cite{Maestre2014}, can additionally be exploited to implement the corresponding MPC problem \cite{Sturz2020,Hu2020} and the posterior computation and sampling in a scalable distributed fashion.
  \end{remark}


\section{ANALYSIS}\label{sec:analysis}
  Since the Bayesian MPC algorithm can conceptually be used to enhance any existing MPC application for a wide variety parametrized transition and reward functions~\eqref{eq:transition} and \eqref{eq:reward}, we briefly recap general sufficient conditions from \cite{Wabersich2020} in this section to obtain finite-time regret bounds that are based on the framework proposed by~\cite{Osband2014}. Part of the main contribution of this paper  compared to~\cite{Wabersich2020} is to establish that these conditions hold for the relevant case of linear systems in Section~\ref{subsec:regularity_linear_concave} and to provide
  sufficient conditions for the more general nonlinear case in
  Section~\ref{subsec:regularity_nonlinear}.
  These results then enable us together with cautious soft constraints from Section~\ref{subsec:soft_mpc}
  to derive a bound on the expected number of unsafe learning episodes under application
  of the Bayesian MPC algorithm in Section~\ref{subsec:unsafe_learning_episodes}.

  We start by reviewing the main steps of model-based RL based on posterior
  sampling arguments as presented in~\cite{Osband2014} to reformulate
  the regret in terms of the expected learning progress
  of the transition and reward function. By using a regularity assumption on
  the expected future reward under the sampled MPC controllers this then allows
  us to bound the cumulative regret using the so-called Eluder dimension,
  which expresses the learning complexity for different mean and reward function
  classes.
  Instead of the instant regret $\mRe{e}$ in~\eqref{eq:instant_regret},
  which includes the unknown optimal future reward $\mV{\mPa}{\mPa}{0}(x)$,
  we formulate the regret in terms of the sampled MPC
  controller applied to the corresponding sampled system,
  for which it is optimal, i.e.
  \begin{align*}
    \mExpectationExpSmall{\mPa, \mPa_{e}, x, \mData_e }{\mtRe{e}}
      = \mExpectationExpSmall{\mPa, x, \mData_e}{
          \mExpectationExpSmall{\mPa_{e}}{
              \mV{\mPa_e}{\mPa_e}{0}(x)
            -
              \mV{\mPa}{\mPa_e}{0}(x)
            |
            \mPa, x, \mData_e
          }
        },
  \end{align*}
  where $\mV{\mPa_e}{\mPa_e}{0}(x)$ is known based on the sample $\mPa_e$ and
  $\mV{\mPa}{\mPa_e}{0}(x)$ can be observed.
  It holds
  $
    \mExpectationExpSmall{\mPa, \mPa_{e}, x, \mData_e}{\mRe{e} - \mtRe{e}} =0
    \Rightarrow
    \mExpectationExpSmall{\mPa, \mPa_{e}, x, \mData_e}{\mRe{e}} = \mExpectationExpSmall{\mPa, \mPa_{e}, x, \mData_e }{\mtRe{e}},
  $
  which allows another reformulation based on the Bellman operator as proposed by~\cite{Osband2014} for the continuous case to end up with a regret bound of the form
  \begin{align}\nonumber
      \mExpectationSmall{\mtRe{e}}  \leq 
      &\biggl[
          \sum_{t=0}^{T-1}
          \mathbb E\bigg[
            | \mV{\mPa_e}{\mPa_e}{t+1}(f(x(t),u(t);\mPa_e)+\epsilon_F(t)) \\\nonumber
            & \qquad\quad~ - \mV{\mPa_e}{\mPa_e}{t+1}(f(x(t),u(t);\mPa)+\epsilon_F(t))|
          \bigg]
      \biggl]
      \\\label{eq:two_rotated regret terms}
      &+\mExpectationExp{}{
          \sum_{t=0}^{T-1}
          | r(t,x(t), u(t);\mPa_e)-r(t,x(t), u(t);\mPa) |},
  \end{align}
  with expectation over $\mPa$, $\mPa_e$, $x$, $\mData_e$, $\epsilon_F(t)$, and $\epsilon_R(t)$.
  The second term in~\eqref{eq:two_rotated regret terms} can be
  bounded via the conditional posterior through
  $
    \mExpectationExpSmall{\mPa, x, \mData_e }{
      \mExpectationExpSmall{\mPa_{e}}{
          \sum_{t=0}^{T-1}
          | r(t, x(t), u(t);\mPa_e) - r(t, x(t), u(t);\mPa)|
          ~ | ~ \mPa, x, \mData_e
      }
    }.
  $
  The first term, however, requires a regularity assumption that
  quantifies how errors in the expected one-step-ahead prediction
  $f(x(t),u(t);\mPa_e) - f(x(t),u(t);\mPa)$ cause deviations w.r.t.
  the one-step-ahead expected future reward $\mV{\mPa_e}{\mPa_e}{t+1}(.)$:
  \begin{assumption}\label{ass:continuity_v}
    For all $\mPa_e\in \RR^{n_\mPa}$ and $x^+,\tilde x^+\in\RR^n$ there
    exists a constant $L_V>0$ such that
    \begin{align}\nonumber
      &\mExpectationExp{\epsilon_F(t)}{
            | \mV{\mPa_e}{\mPa_e}{t+1}(x^+ + \epsilon_F(t)) 
              - \mV{\mPa_e}{\mPa_e}{t+1}(\tilde x^+ + \epsilon_F(t))|
          } \\
      & ~~ \leq L_V \Vert x^+ - \tilde x^+\Vert _2. \label{eq:ass_continuity_v}
    \end{align}
  \end{assumption}
  While previous literature only required this assumption, we will theoretically investigate its justification in the case of MPC-based policies in Section~\ref{subsec:regularity_linear_concave} and Section~\ref{subsec:regularity_nonlinear}.

  The relationship between the regret and the mean deviation between
  the true and sampled reward and transition function as described previously
  allows us to derive a Bayesian regret bound using statistical measures.
  The first bound relates to the complexity of the respective mean
  function also known as the Kolmogorov dimension $\dim_K$,
  see also \cite{Russo2014}. In an online learning setup it is additionally
  necessary to quantify the difficulty of extracting information and accurate
  predictions based on observed data, which is measured in terms of the
  Eluder dimension $\dim_E$ \cite{Russo2014}. These measures further
  require boundedness of the mean reward and transition function as
  follows.\footnote{ Most real-world systems are bounded, e.g., due to finite energy or limited raw materials, justifying Assumption~\ref{ass:bounded_expected_cost_transition} in practice. A modified LQR setting to match Assumption~\ref{ass:bounded_expected_cost_transition} can be found, e.g., in~\cite[Appendix D]{Osband2014}.}
  \begin{assumption}\label{ass:bounded_expected_cost_transition}
    There exist constants $c_\mR$ and $c_F$ such that
    for all admissible $x\in\RR^n$, $u\in\RR^m$, $\mPa\in\RR^{n_{\mPa}}$,
    and $t=0,1,..,T-1$ it holds
    $|\ell(t,x,u;\mPa_\mR)|\leq c_\mR$, and $||f(x,u;\mPa_F)||\leq c_F$.
  \end{assumption}
  Following~\cite{Osband2014,Wabersich2020}, we can combine these measures
  to obtain the following regret bound for the Bayesian MPC
  algorithm as an immediate consequence from Theorem 1 in~\cite{Osband2014}
  with $\tilde \OO$ neglecting terms that are logarithmic in $N$.
  \begin{theorem}\label{thm:general_regret_bound}
    If Assumptions~\ref{ass:continuity_v} and \ref{ass:bounded_expected_cost_transition}
    hold then it follows that
    \begin{align}\nonumber
      \mCre(N)
      \leq 
      \tilde \OO\bigg(&
        \sigma_R\sqrt{\mKdim(\ell)\mEdim(\ell)T N}
        \\&
        +L_V \sigma_F \sqrt{\mKdim(f)\mEdim(f)T N}
      \bigg).\label{eq:cor_general_regret_bound}
    \end{align}
  \end{theorem}
  Specific bounds for different parametric function classes can be
  found, e.g., in~\cite{Russo2014,Osband2014}.

\subsection{Regularity of the value function for large-scale linear transitions and concave rewards}
\label{subsec:regularity_linear_concave}
  Regularity of the future reward as required by Assumption~\ref{ass:continuity_v}
  is a central ingredient for the performance analysis and essentially determines
  the shape of the regret bound in Theorem~\ref{thm:general_regret_bound}.
  While explicit bounds on the Kolmogorov- and Eluder dimensions are available
  for relevant parametric function classes~\cite{Russo2014,Osband2014}, we provide
  a bound on $L_V$ that holds under application of the Bayesian MPC algorithm.
  In this section we begin by focusing on the control relevant case of
  linear time-invariant transitions~\eqref{eq:transition} of the form
  \begin{align}\label{eq:linear_transition}
    x(t+1) = A(\theta_F)x(t) + B(\theta_F)u(t) + \epsilon_F(t)
  \end{align}
  and reward models~\eqref{eq:reward} that are either affine or quadratic and concave
  in the states and inputs for each time step $t=0,1,..,T-1$.
  Furthermore, we restrict our attention to state and input spaces
  that are polytopic of the form
  $\XX\mDef\{x\in\RR^n | A_x x \leq b_x\}$ and $\UU\mDef\{u\in\RR^m | A_u u \leq b_u\}$.
  Based on these assumptions we establish Lipschitz continuity of the optimizer
  of the MPC Problem~\eqref{eq:mpc_problem}.
    \begin{theorem}\label{thm:lipschitz_continuity_linear}
      Consider MPC problem \eqref{eq:mpc_problem}. If
      the mean transition~\eqref{eq:mpc_problem_transition} is linear,
      the state~\eqref{eq:mpc_problem_state_constraint} and
      input~\eqref{eq:mpc_problem_action_constraint} constraints are polytopic, and
      the mean reward function~\eqref{eq:mpc_problem_reward} is
      linear or quadratic and strictly concave for all time steps, then
      under application of the Bayesian MPC algorithm it follows
      that Assumption~\ref{ass:continuity_v} holds.
    \end{theorem}
    The proof together with a detailed construction of $L_V$ according to
    Asssumption~\ref{ass:continuity_v} can be found in
    Appendix~\ref{app:proof_lipschitz_continuity_linear}.
    
    By combining this result with Theorem~\ref{thm:general_regret_bound}
    and the specific bounds on the Eluder- and Kolmogorov dimensions from~\cite{Russo2014,Osband2014},
    we obtain the specific bound
    \begin{align}\label{eq:specific_regret_bound}
      \mCre(N)
      &\leq
      \tilde \OO\left(
        \sigma_{R}n_\ell\sqrt{2 T N}
        +L_V \sigma_{F} n\sqrt{n(n+m) T N}
      \right)
    \end{align}
    with $L_V$ according
    to \eqref{eq:lipschitz_continuity_value_function} in 
    Appendix~\ref{app:proof_lipschitz_continuity_linear}.
\subsection{Extension towards nonlinear transitions and rewards}\label{subsec:regularity_nonlinear}
  While the linear case as considered in the previous section covers
  a large portion of control applications, the increasing availability and performance
  of nonlinear MPC solvers motivates the extension to nonlinear reward
  and transition models. We therefore extend the analysis from
  Section~\ref{subsec:regularity_linear_concave} to the more general case
  of transition and reward functions $f$ and $\ell$
  that are nonlinear and non-convex as well as more general state and input
  spaces of the form $\XX\mDef\{x\in\RR^n | g_x(x) \leq \mathbb 1\}$ and 
  $\UU\mDef\{u\in\RR^m | g_u(u) \leq \mathbb 1\}$.
  To this end, we use a similar line of reasoning as in the proof of Theorem~\ref{thm:lipschitz_continuity_linear} to provide sufficient conditions on the resulting nonlinear MPC problem~\eqref{eq:mpc_problem} that ensure Assumption~\ref{ass:continuity_v}. In particular, we utilize results from~\cite{Liu1995} to analyze local continuity properties of KKT-based solutions of the MPC problem~\eqref{eq:mpc_problem}, which results in sufficient conditions commonly used in nonlinear optimization algorithms~\cite{wachter2006implementation}, and in the nonlinear MPC literature~\cite[Chapter 8.6.1]{Rawlings2017}\re.
  \begin{theorem}\label{thm:lipschitz_continuity_nonlinear}
    Let Assumption~\ref{ass:bounded_expected_cost_transition} hold and consider
    the MPC problem~\eqref{eq:mpc_problem} with
    $f,\ell,I,g_x$, and $g_u$ continuously differentiable and Lipschitz continuous.
    If the linear independence (LI) and the strong second-order sufficient
    condition (SSOSC) according to \cite[3(a) and 3(d)]{Liu1995} hold for all
    admissible $x\in\RR^n$, $\mPa\in\RR^{n_{\mPa}}$, and $t=0,1,..,T-1$,
    then it follows that Assumption~\ref{ass:continuity_v} holds.
  \end{theorem}
  The main steps of the proof can be found in Appendix~\ref{app:proof_lipschitz_continuity_nonlinear}. 
  The linear independence (LI) condition refers to the linear
  independence of the gradients of the active constraints at an
  optimum of~\eqref{eq:mpc_problem} with respect to the decision variables
	and ensures necessity of the corresponding KKT conditions.
	In addition, the strong second-order sufficient condition (SSOSC)
	guarantees sufficiency of the KKT conditions and uniqueness of local solutions
	through a local positive definiteness condition of the Hessian of
	the Lagrangian with respect to the decision variables, also depending on the
	active constraints at the optimum.
  Consequently, if these conditions do not hold,
	situations where small deviations of $x$ cause a `jumping behavior' between
  different local optimal solutions of the MPC
  problem~\eqref{eq:mpc_problem} can occur and potentially lead to a non-Lipschitz
  continuous future reward function.
  While the imposed assumptions are difficult to verify, note that the LI condition is a common requirement for nonlinear solvers and that the SSOSC is the weakest condition to ensure existence and uniqueness of local solutions of the MPC problem~\eqref{eq:mpc_problem} for small perturbations of the initial condition $x$, see \cite{Kojima1980}.  In addition, recent results connect these assumptions and the resulting Lipschitz continuity property to the actual underlying optimal control problems~\cite{dontchev2019lipschitz}. Importantly, note that, e.g., a normally distributed $\epsilon_F$ helps to smoothen the future reward through the expectation operator in~\eqref{eq:ass_continuity_v},  rendering the conditions for Bayesian MPC less restrictive than common assumptions from the nonlinear MPC literature.

\subsection{Bounding the expected number of unsafe learning episodes}
\label{subsec:unsafe_learning_episodes}
  While a tightened MPC formulation using the true parameters $\mPa$ typically provides state constraint satisfaction in expectation for many practical applications, the parameter samples $\mPa_e$ during application of the Bayesian MPC algorithm can vary significantly during initial learning episodes. It can therefore happen that the constraints are violated, even in expectation. In such learning episodes, however, the amount of constraint violation can partially be observed through the regret due to the exact penalty in the future reward function~\eqref{eq:reward}. As a consequence, if the MPC using the true system parameters $\mPa$ provides satisfaction of the
  tightened constraints in expectation, i.e. $x(t)\in\bar\XX_k(0)$, we can use regularity of the expected future reward to show that a converging parameter estimate yields a converging future reward and therefore converging constraint satisfaction. In other words, since the stage cost function is bounded and the constraints are tightened, a sufficiently large soft constraint penalty ensures observability and a bound on state constraint violations, which can be formalized as follows.
  We first derive an upper bound on the instant regret $\Delta_e$ as defined in~\eqref{eq:cumulative_regret} implying constraint satisfaction in Appendix~\ref{app:safety_analysis}. By combining this intermediate result with the regret bound from Theorem~\ref{thm:general_regret_bound} we bound the cumulative expected number of unsafe learning episodes in Theorem~\ref{thm:total_number_of_unsafe_episodes}. To streamline notation we denote the state, input, and slack variable sequence in the expected future reward~\eqref{eq:reward} for a given initial state $x$ as $x_{\tilde\mPa}^\mPa(j)$, $u_{\tilde\mPa}^\mPa(j)$, and $\rho_{\tilde\mPa}^\mPa(j)$ for $j=0,..,T-1$ in the following.
    
  \begin{theorem}\label{thm:total_number_of_unsafe_episodes}
    Let the conditions of Theorem~\ref{thm:general_regret_bound} hold and
    consider a weighting factor in the exact penalty term $I(.)$
    in~\eqref{eq:mpc_problem_reward} that satisfies
    $\min_i(c_{1,i}) \geq \frac{2Tc_R + c_\delta}{\delta_0}$ with $c_R$ from
    Assumption~\ref{ass:bounded_expected_cost_transition} and some $c_\delta >0$.
    If $\mExpectationExp{E,x}{\rho_{\mPa}^\mPa(j)}=0$, then the total number of
    $N_{\mathrm{unsafe}}$ episodes, for which there exists a
    $j \in\mathbb N$, $0\leq j \leq T-1$ such that $\mExpectationExp{E,x}{x_{\mPa_e}^\mPa(j)\notin\XX}$
    is bounded in terms of the cumulative regret by
    $N_{\mathrm{unsafe}}\leq\lceil CR(N)c_\delta^{-1}\rceil$.
  \end{theorem}
  A sublinear cumulative regret bound therefore ensures a decreasing ratio between the number of episodes with constraint violation and the total number of learning episodes, which vanishes at the rate of $c(1/N)$ for $N \rightarrow \infty$ and some positive constant $c$.
  
  It should be noted that Theorem~\ref{thm:total_number_of_unsafe_episodes} requires that $\min_i(c_i)\geq 1/\delta_0$ holds up to linear factors, which implies that a smaller constraint tightening $\delta_0$ results in a larger constraint violation penalty $I(\rho)$ in~\eqref{eq:reward} by the factor $1/\delta_0$ in case of unsafe episodes. While this does not affect the overall structure of the bound and its dependence on $\sqrt N$ in Theorem~\ref{thm:general_regret_bound}, a practical choice of $\delta_0$ in the related MPC literature is typically between $0.005$ and $0.1$.
  While the slack penalties in~\eqref{eq:reward} can similarly be found in constrained RL concepts~\cite{Wei2021,Liu2021,Efroni2020},  it should be noted that the proposed MPC policy directly minimizes the slacks through the soft constraints~\eqref{eq:mpc_problem_state_constraint} and penalty~\eqref{eq:mpc_problem_reward} and constraint violations only occur due to model uncertainties. In contrast, policy optimization methods~\cite{Wei2021,Liu2021,Efroni2020}, e.g., perform primal and dual update steps after each episode, resulting in separated bounds on performance and constraint satisfaction. 
  Note that the upper bound $c_\delta$ and consequently also the lower bound on the exact penalty scaling could potentially be improved in the corresponding proof (Appendix~\ref{app:safety_analysis}) by exploiting the concrete structure of $I$ including quadratic terms.
  
  Furthermore, for practical applications with large initial parametric uncertainties, one could consider the use of a more restrictive constraint tightening $\bar \delta_0>\delta_0$ during a finite number of initial training episodes to reduce the number of constraint violations within the initial learning phase.

\section{NUMERICAL EXAMPLES}\label{sec:numerical_examples}
  \begin{figure*}
    \centering
    \includegraphics[height=4.55cm]{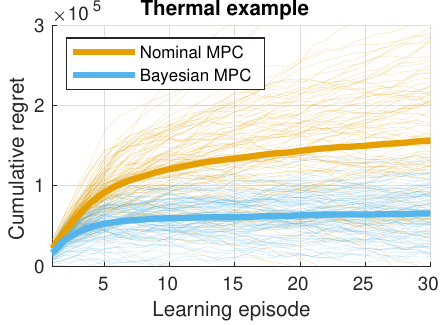}
    \hfill
    \vline
    \hfill
    \includegraphics[height=4.5cm]{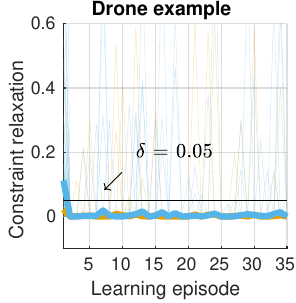}
    \hfill
    \includegraphics[height=4.5cm]{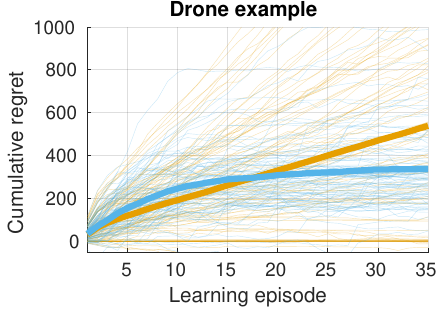}
    \caption{Simulation results of numerical examples for 100 different experiments.
      Thin lines depict experiment samples and thick lines show the corresponding
      mean with brown and blue lines indicating nominal posterior MPC and the
      proposed Bayesian MPC.
      \textbf{Left:} Cumulative regret of the large-scale thermal application
      detailed in Section~\ref{subsec:cooling_example}.
      \textbf{Middle:} Maximum value of $\rho(j)$ as defined in~\eqref{eq:reward}
      over one episode. \textbf{Right:} Cumulative regret of exploration task
      as described in Section~\ref{subsec:drone_example}.}
    \label{fig:example_plots}
  \end{figure*}
  All examples are implemented using the IPOPT solver~\cite{wachter2006implementation}. The source code of the numerical examples can be found online\footnote{http://www.kimpeter.de/wp-content/uploads/2022/09/code.zip}
  \subsection{Large scale thermal application}\label{subsec:cooling_example}
  \begin{figure}
    \centering
    \includegraphics[width=0.85\linewidth]{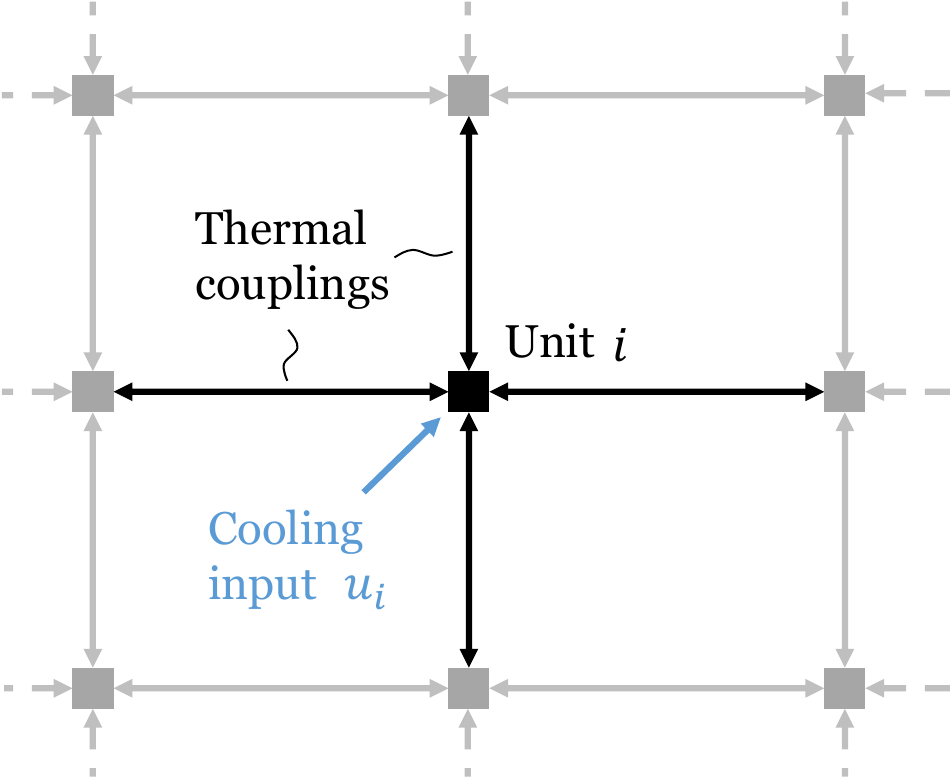}
    \caption{Cooling network structure: Each unit $i$ has a measured
    temperature state $x_i$ and cooling input $u_i$ and is affected by
    the temperature states of the top/down/left/right neighboring units
    arranged in a grid.}
    \label{fig:experiment_cooling}
  \end{figure}
  We first consider the task of efficiently controlling a
  large-scale network with 100 cooling units, e.g. a server farm or
  production machines in manufacturing plants, that are
  arranged in a grid structure and have a strong
  thermal coupling with respect to locally neighboring units,
  see Figure~\ref{fig:experiment_cooling}.
  Inputs $u(t)\in\UU\subset\RR^{100}$ describe the
  applied cooling power to each cooling, which are subject
  to physical limitations $\UU\mDef\{u\in\RR^{100}|u\leq \bar u\}$.
  The system state is defined by the temperatures $x(t)\in\RR^{100}$ of
  each unit that needs to be below a given threshold
  $x(t)\in\XX\mDef\{x\in\RR^{100}|x\leq 100\}$ for all times.
  The thermodynamics of the plant are given with linear mean dynamics
  such that each unit $i$ has unknown dynamics of the form 
  $x_i(t+1)= A_{ii}x(t) + B_{i}u_i(t) + \sum_{j\in N_{s,i}}A_{ij}x_j(t) + C_i + \epsilon_F(t)$ with neighboring
  units indexed by $j$, known Gaussian parameter prior distribution $(A,B,C)\sim\mDistribution{\mPa_F}$
  and Gaussian process noise $\epsilon_F(t)$. Note that the resulting overall
  dynamics can be stated in the form~\eqref{eq:linear_transition} by extending the state space.
  The goal is to minimize the overall expected energy consumption $\sum_i L_i u_i(t) + \epsilon_\mR$
  by considering thermal couplings while keeping the temperature of each cooling unit below
  a specified maximum temperature, starting form a temperature level below $100$
  degrees. The energy efficiency of each unit is described through parameters $L\in\RR^{100}$
  that are sampled from a known Gaussian prior distribution $l\sim\mDistribution{\mPa_\mR}$
  plus additive Gaussian measurement noise $\epsilon_\mR$. 
   Based on these assumptions, i.e., the distributed linear structure of the dynamics and reward in combination with Gaussian prior distribution, each unit $i$ can compute its posterior distribution $\mDistribution{\mPa_i|\mData_{e,i}}$ and corresponding samples independently from other units in closed-form~\cite[Section 3.2.1]{hewing2020}.
  The overall plant consists partly of new cooling units with known efficiency and
  older cooling units with uncertain efficiency factors that are worse in expectation.
  Due to these different efficiency levels that are provided
  through the prior distribution, explorative behavior can be beneficial
  to exploit more efficient units.
  The exact numerical values and prior parametrisation can be found in
  the function $\texttt{server\_experiment.m}$ in the provided source code
  for the example.
  In Figure~\ref{fig:example_plots} (Left), we compare the proposed
  Bayesian MPC algorithm against commonly used nominal posterior
  MPC, i.e. selecting $\theta_e \mDef \mExpectation{\theta|\mData_e}$ 
  \cite[Section 3]{hewing2020}, using 100 different system realizations.
  While both algorithms show reasonable learning performance and provide
  constraint satisfaction at all times, Bayesian MPC is
  able to significantly reduce the cumulative regret by almost $50\%$
  compared to nominal posterior MPC.

\subsection{Drone search application}\label{subsec:drone_example}
  In this example, we consider a generic drone search task falling into the problem class of Section~\ref{subsec:regularity_nonlinear}. The goal is to collect information about an a-priori unknown position of interest using a quadrotor drone. While the prior of the 10-dimensional drone dynamics are selected according to~\cite{bouffard2012board}, we additionally simulate strong winds in different altitudes, which adds strong nonlinear effects to the dynamics. Once the target position is reached, the drone collects information before it returns to the base station for analysis and recharge. The overall goal therefore is to learn the drone dynamics, winds in different altitudes and the most informative search position. The safety-critical constraints are a maximum range of the drone together with a minimum altitude that need to be satisfied under physical actuator limitations.
  The dynamics are of the form $x(t+1)=Ax(t) + Bu(t) + C\Phi(x(t)) + \epsilon_F$,
  where $A$ and $B$ matrices describe the unknown dynamics around the hovering state and 
  and $C\Phi(x(t))$ models strong winds in different altitudes as
  $\Phi(x(t))=[k(x_3(t),\bar W_1),k(x_3(t),\bar W_2),k(x_3(t),\bar W_3)]^\top$
  using radial basis functions $k(.,.)$~\cite{Bishop2006} with given hyper-parameters
  $W_i$ and unknown parameters $C$.
  The system has a three dimensional input space that allows to control the desired pitch
  and roll as well as vertical acceleration of the drone.
  The system constraints are given by physical input constraints and a box constraint
  on the position states, which describes the minimal altitude and maximum range of the drone.
  Furthermore, the use of a linear model is only valid around the hovering state, yielding additional
  absolute pitch and roll constraints of $30~\mathrm{[deg]}$ to the system.
  The reward signal corresponds to the information gained at the final position
  at the end of an episode and is modeled as the sum of equally spread radial-basis-functions
  $\kappa(.,.)$ at positions $p_i\in\RR^3$, $i=1,..,9$ , i.e.
  $r(T-1, x(T-1), u(T-1)) = \epsilon_\mR + \sum_{i=1}^9 \theta_{R,i} \kappa(_{1-3}(T-1), p_i)$.
  The unknown system dynamics parameters $\mPa_F\mDef C$, dynamics process noise $\epsilon_F$,
  reward parameters $\mPa_\mPa\mDef \{p_i\}$, and reward noise $\epsilon_\mR$ are normally distributed  and allow to compute the posterior distribution in closed-form~\cite[Section 3.2.1]{hewing2020}.
  The exact numerical values can be found in the function \texttt{quadrotor\_example.m} in
  the provided source code for the example. 
  By comparing nominal posterior MPC against Bayesian
  MPC over 100 different experiments in terms of expected
  constraint satisfaction, we notice from Figure~\ref{fig:example_plots} (Middle) that
  Bayesian MPC causes explorative behavior during initial episodes,
  which yields higher constraint violations compared to nominal posterior MPC. However,
  this behavior enables safety of future episodes and bounded
  cumulative regret Figure~\ref{fig:example_plots} (Right) compared to
  posterior nominal MPC, which has unbounded cumulative regret.

\section{CONCLUSION}
  In this paper, we combined model predictive control with reinforcement learning based on
  posterior sampling in an episodic setting to efficiently learn an optimal MPC control policy
  for dynamical systems that are subject to state and input constraints.
  Using arguments from multi-parametric optimization, we were able to 
  rigorously bound the learning performance of the resulting Bayesian MPC algorithm
  in the case of linear mean transition functions, concave mean rewards, and polytopic
  system constraints.
  For more general systems that can be described through Lipschitz
  continuous nonlinear functions, we derived sufficient conditions
  for bounded cumulative learning regret using sensitivity analysis results,
  providing insights into the general applicability of Bayesian MPC.
  To account for external disturbances as well as large
  parametric uncertainties during initial learning episodes,
  we introduced softened state constraints that are iteratively
  tightened along the prediction horizon, guaranteeing a sublinear bound
  on the number of expected unsafe learning episodes with state constraint
  violations.
  While the proposed algorithm maintains the online computation load
  of nominal MPC, the advantage in terms of
  learning performance was demonstrated in simulation using a
  large-scale thermal control problem together with a highly nonlinear
  drone application.

\bibliographystyle{unsrt}
\bibliography{bibliography.bib}

\appendix

\section{PROOFS}
  \begin{figure*}[t]
      \vspace{0.2cm}
      \hrulefill
      \vspace{0.2cm}
      \begin{align}
          |\mV{\mPa_e}{\mPa_e}{t+1}(x^+) 
                - \mV{\mPa_e}{\mPa_e}{t+1}(\tilde x^+)|
          =&
          \left|\mExpectationExp{E}{
            \sum_{j=t+1}^{T-1} \tilde\ell(j, x(j,x^+,E), u(j,x(j,x^+,E)))
              - \tilde\ell(j, x(j,\tilde x^+,E), u(j,x(j,\tilde x^+,E)))
          }\right|\nonumber\\
          \leq&
          \mExpectationExp{E}{
            \sum_{j=t+1}^{T-1}
              \left|\tilde\ell(j, x(j,x^+,E), u(j,x(j,x^+,E)))
              - \tilde\ell(j, x(j,\tilde x^+,E), u(j,x(j,\tilde x^+,E)))
              \right|
          }\nonumber\\
          \leq&
          \mExpectationExp{E}{
            \sum_{j=t+1}^{T-1}
              L_{\ell,j}||[x(j,x^+,E)^\top-x(j,\tilde x^+,E)^\top,
                u(j,x(j,x^+,E))^\top-u(j,x(j,\tilde x^+,E))^\top]^\top||
          }\label{eq:app_suff_lipschitz}
        \end{align}
      \vspace{0.2cm}
      \hrulefill
      \vspace{0.2cm}
      \begin{align}
        |\mV{\mPa}{\mPa}{0}(x) - \mV{\mPa}{\tilde\mPa}{0}(x)|
        & \geq \left|\mExpectationExp{E}{\ell_{\mPa}^{\mPa}(\bar j) - \ell_{\tilde\mPa}^{\mPa}(\bar j) + \min_i(c_{1,i})\delta_0 +  
          \sum_{j\in\{0,...,T-1\}\setminus\{\bar j\}} \ell_{\mPa}^{\mPa}(j) - \ell_{\tilde\mPa}^{\mPa}(j) + I(\rho_{\tilde \mPa}^{\mPa}(j))}\right|\nonumber\\
        & \geq \left|\mExpectationExp{E}{(T-1)2c_R + c_\delta + \sum_{j\in\{0,...,T-1\}\setminus\{\bar j\}} \ell_{\mPa}^{\mPa}(j) - \ell_{\tilde\mPa}^{\mPa}(j) + I(\rho_{\tilde \mPa}^{\mPa}(j))}\right| \nonumber\\
        & \geq \left|\mExpectationExp{E}{(T-1)2c_R + c_\delta + \sum_{j\in\mathbb J\setminus\{\bar j\}} \ell_{\mPa}^{\mPa}(j) - \ell_{\tilde\mPa}^{\mPa}(j) + I(\rho_{\tilde \mPa}^{\mPa}(j))
          + \sum_{j\notin\mathbb J\setminus\{\bar j\}} \ell_{\mPa}^{\mPa}(j) - \ell_{\tilde\mPa}^{\mPa}(j) + I(\rho_{\tilde \mPa}^{\mPa}(j))}\right|\label{eq:app_lemma_exp_episodes}
      \end{align}
      \vspace{0.2cm}
      \hrulefill
      \vspace{0.2cm}
    \end{figure*}
  \subsection{Proof of Theorem~\ref{thm:lipschitz_continuity_linear}}
  \label{app:proof_lipschitz_continuity_linear}
    It is sufficient to show global Lipschitz continuity of
    $\mV{\mPa_e}{\mPa_e}{t}(x)$ since existence of an $L_{V_t}>0$ such
    that
    \begin{align*}
      &\mExpectationExp{\epsilon_F}{
          | \mV{\mPa_e}{\mPa_e}{t+1}(x^+ + \epsilon_F) 
            - \mV{\mPa_e}{\mPa_e}{t+1}(\tilde x^+ + \epsilon_F)|
        }\\
      &\leq
      \mExpectationExp{\epsilon_F}{
          L_{V_t}||x^+ + \epsilon_F-\tilde x^+ + \epsilon_F||
        }\\
      &\leq
      L_{V_t}||x^+ -\tilde x^+ ||
    \end{align*}
    implies the desired result.
    Let $\tilde\ell(j,x,u)\mDef \ell(j,x,u;\mPa_e) - I(\rho(x))$ with
    $\rho(x)\mDef \min_{\rho\geq 0}\text{ s.t. }s\in\bar\XX_k(\rho)$ according
    to~\eqref{eq:reward}, which is Lipschitz continuous in $x$ and $u$ since
    $\XX$ and $\UU$ are polytopic.
    To streamline notation we denote the state and input sequence
    in the expected future reward~\eqref{eq:reward} with $\tilde \mPa = \mPa_e$
    as $x(j, x, E)$ and $u(j, x)$ for $j=0,..,T-1$ with $x(0, x, E)=x$ and
    $E \mDef [\epsilon_{\mR}(t), .., \epsilon_{\mR}(T-1),\epsilon_{F}(t), .., \epsilon_{F}(T-2)]$
    in the following.
    We have due to linearity of the expectation operator, Jensen's inequality,
    the triangle inequality, and global Lipschitz continuity of $\tilde\ell$ that
    there exists $L_{\ell,j}$ such that~\eqref{eq:app_suff_lipschitz} holds.
    It therefore remains to show that $x$ and $u$ are Lipschitz continuous in
    their second argument.
    Since $u(j,x)$ is the first element of the optimal input sequence
    according to~\eqref{eq:mpc_problem} and~\eqref{eq:mpc_problem} 
    is guaranteed to be feasible due to the soft-constraint reformulation,
    it follows from \cite[Thm. 1.8]{borrelli2003constrained} for
    affine $\tilde\ell$ and from \cite[Thm. 1.12]{borrelli2003constrained} for
    strictly concave quadratic $\tilde\ell$ for any $j$
    and $x,\tilde x\in\RR^n$ that there exists
    a $\bar K \in \RR^+$, $\bar K <\infty$ such that
    \begin{align}
       ||u(j,x)-u(j,\tilde x)|| \leq \bar K ||x-\tilde x||.\re
    \end{align}

    It remains to show that there exists an $L_x(j)$ such that
    \begin{align}
      ||x(j,x^+,E)-x(j,\tilde x^+,E)||\leq L_x(j) ||x^+-\tilde x^+||.
    \end{align}
    We have that 
    \begin{align*}
      &||x(j,x^+,E)-x(j,\tilde x^+,E)||\leq L_x(j)||x^+ - \tilde x^+||\\
      \Rightarrow & ||x(j+1,x^+,E)-x(j+1,\tilde x^+,E)||\leq L_x(j+1)||x^+ - \tilde x^+||
    \end{align*}
    with $L_x(j)=L_x(j-1)(||A|| + ||B||\bar K)$ and $L_x(0)=1$ by induction.
    
    Induction start $j = 0$:
    \begin{align*}
      ||x(0,x^+,E)-x(0,\tilde x^+,E)||=
      ||x^+ - \tilde x^+||\leq L(0)||x^+ - \tilde x^+||
    \end{align*}
    with $L(0)=1$ implying
    \begin{align*}
      &||x(1,x^+,E)-x(1,\tilde x^+,E)||\\
        =&||Ax^+ + B u(0, x^+) + \epsilon_F(0) -A\tilde x^+ - B u(0, \tilde x^+) - \epsilon_F(0)||\\
        \leq & \underbrace{(||A||+||B||\bar K)1}_{=L_x(1)}||x^+ - \tilde x^+||.
    \end{align*}
    Induction step for any $j>0$ it holds
    \begin{align*}
      ||x(j,x^+,E)-x(j,\tilde x^+,E)&\leq L_x(j)||x^+ - \tilde x^+||.
    \end{align*}
    We have
    \begin{align*}
      &||x(j+1,x^+,E)-x(j+1,\tilde x^+,E)||\\
    =&||Ax(j,x^+,E) + B u(j, x(j,x^+,E)) + \epsilon_F(j) \\
     &~~-Ax(j,\tilde x^+,E) - B u(j, x(j,\tilde x^+,E)) - \epsilon_F(j)||\\
    \leq&(||A||+||B||\bar K)||x(j,x^+,E)- x(j,\tilde x^+,E)||\\
    \leq&\underbrace{(||A||+||B||\bar K)L(j)}_{L_x(j+1)}||x^+ - \tilde x^+||\text{ (induction hypothesis)}.
    \end{align*}
    Combining these results yields
    \begin{align}\label{eq:lipschitz_continuity_value_function}
      L_V = \sum_{j=t+1}^{T-1} L_{\ell,j}(L_x(j)(1 + \bar K))<\infty,
    \end{align}
    which completes the proof.
    \qed

\subsection{Proof outline of Theorem~\ref{thm:lipschitz_continuity_nonlinear}}\label{app:proof_lipschitz_continuity_nonlinear}
  \begin{figure*}[t]
      
    \end{figure*}
    For any initial state $x_0$ there exists a corresponding
    optimal solution $\mUpredOpt_{t|t}$ to~\eqref{eq:mpc_problem} due to
    the soft-constraint formulation, Lipschitz continuity of the objective,
    Lipschitz continuity of the constraints in~\eqref{eq:mpc_problem},
    and the compactness of the input constraints.
    Together with~\cite[Theorem 3.7]{Liu1995} it follows from the given
    assumptions that there exists a unique function
    $y(t,x) \mapsto [\mUpredOpt_{t|t},\mUpredOpt_{t+1|t},...,\mUpredOpt_{t+N-1|t}]^\top$
    that is Lipschitz continuous with respect to all initial conditions $x\in \BB(x_0,r)$
    with $\BB(x_0,r)\mDef\{s\in\RR^n |~||x-x_0||\leq r\}$
    and $x_0$ fulfilling the KKT conditions corresponding to~\eqref{eq:mpc_problem}.
    Due to the SSOSC, the KKT conditions imply optimality of $y(t,x)$ and we
    conclude existence of a local Lipschitz constant $L_u(t,x_0)>0$ such
    that for all $x,~\tilde x\in\BB(r,x_0)$ it holds
    $||u(t,x) - u(t,\tilde x)||=||y_t(x) - y_t(\tilde x)||\leq L_u(t,x_0)||x - \tilde x||$.
    Since the input space is compact it also follows boundedness of
    \begin{align}	
      L_{\bar u} = \max_{x,\tilde x, ||x-\tilde x||>r} \frac{||u(t,x) - u(t,\tilde x)||}{||x-\tilde x||},
    \end{align}			
    allowing us to select $\bar K\mDef\max\{L_{\bar u},L_u(t,x_0)\}$.
    From here we can proceed analogously to the proof of
    Theorem~\ref{thm:lipschitz_continuity_linear} using
    $||x(j+1,x^+,E) - x(j+1,\tilde x^+,E)||\leq
    (L_{fs} + L_{fa}\bar K)||x(j,x^+,E) - x(j,\tilde x^+,E)||$
    with $L_{fx}$ and $L_{fu}$ being the Lipschitz constants of
    $f$ with respect to the state $x$ and input $u$.\qed

\subsection{Bounding the expected number of unsafe learning episodes}\label{app:safety_analysis}
  \setcounter{lemma}{0}
  \renewcommand{\thelemma}{\Alph{subsection}.\arabic{lemma}}
  \begin{lemma}\label{lem:reward_constraint}
    Let Assumption~\ref{ass:bounded_expected_cost_transition} hold.
    Consider the expected future reward in~\eqref{eq:reward} for a constraint tightening
    $\delta_{k-t} > \delta_0 >0$ and $x\in\bar\XX_k(0)$. If $\mExpectationExp{E}{\rho_{\mPa}^\mPa(j)}=0$
    for all $j=0,..,T-1$ and the weighting factor of the exact penalty term $I(.)$
    in~\eqref{eq:mpc_problem_reward} satisfies
    $\min_i(c_{1,i}) \geq \frac{2Tc_R + c_\delta}{\delta_0}$ for some $c_\delta >0$
    then it holds
    \begin{align}
      |\mV{\mPa}{\mPa}{0}(x) - \mV{\mPa}{\tilde\mPa}{0}(x)| < c_\delta
      \Rightarrow \mExpectationExp{E}{x_{\tilde\mPa}^\mPa(j)\in\XX}
    \end{align}
    for all $j=0,1,..,T-1$.
  \end{lemma}
  \begin{proof}
    For a proof by contradiction, consider the case
    $| \mV{\mPa}{\mPa}{0}(x) - \mV{\mPa}{\tilde\mPa}{0}(x)| < c_\delta$ 
    and $\mExpectationExp{E}{x_{\tilde\mPa}^\mPa(\bar j)\notin\XX}$ for some
    $0\leq \bar j \leq T-1$.
    It holds $\max_i\mExpectationExp{E}{\rho_{i,\tilde\mPa}^\mPa(\bar j)}>\delta_0$ and
    $\mExpectationExp{E}{I(\rho_{\tilde\mPa}^\mPa(\bar j))}\geq
      \min_i(c_{1,i})\max_i\mExpectationExp{E}{\rho_{i,\tilde\mPa}^\mPa(\bar j)}$.
    Next, we derive a lower bound on the absolute expected reward difference
    \begin{align*}
        |\mV{\mPa}{\mPa}{0}(x) - \mV{\mPa}{\tilde\mPa}{0}(x)|
      &= \left|
        \mExpectationExp{E}{
          \sum_{j=0}^{T-1}\ell_{\mPa}^{\mPa}(j) - \ell_{\tilde\mPa}^{\mPa}(j)
          + I(\rho_{i,\tilde\mPa}^\mPa(j))
        }
        \right|
    \end{align*}
    with $\ell_{\tilde \mPa}^{\mPa}(j)\mDef\ell(j, x_{\tilde\mPa}^\mPa(j), u_{\tilde\mPa}^\mPa(j);\theta_\mR)$
    to show the contradiction. We distinguish two cases:
    
    Case $\mExpectationExp{E}{\ell_{\mPa}^{\mPa}(j) - \ell_{\tilde\mPa}^{\mPa}(j)} \geq 0$
    for all $j$:
    It follows directly that $|\mV{\mPa}{\mPa}{0}(x) - \mV{\mPa}{\tilde\mPa}{0}(x)| \geq \min_i(c_{1,i})
    \max_i\mExpectationExp{E}{\rho_{i,\tilde\mPa}^\mPa(\bar j)}
    \geq c_\delta$.

    Case $\mExpectationExp{E}{\ell_{\mPa}^{\mPa}(j) - \ell_{\tilde\mPa}^{\mPa}(j)} < 0$ with $j\in\mathbb J$ for some
    index set $\mathbb J \subseteq\{0,..,T-1\}$ yields~\eqref{eq:app_lemma_exp_episodes},
    where we use linearity of the expectation operator, $\min_i(c_{1,i}) \geq \frac{2Tc_R + c_\delta}{\delta_0}$,
    and the fact that $\ell_{\mPa}^{\mPa}(\bar j) - \ell_{\tilde\mPa}^{\mPa}(\bar j)> - 2c_R$.
    Similarly $\sum_{j\in\mathbb J\setminus\{\bar j\}} \ell_{\mPa}^{\mPa}(j) -
      \ell_{\tilde\mPa}^{\mPa}(j) \geq -|\mathbb J|c_R \geq - (T-1)2c_R$ and
    $I(\rho_{\tilde \mPa}^{\mPa}(j)) \geq 0$ yielding
    \begin{align*}
      |\mV{\mPa}{\mPa}{0}(x) - \mV{\mPa}{\tilde\mPa}{0}(x)| 
      & \geq \left|c_\delta + \sum_{j\notin\mathbb J\setminus\{\bar j\}} \ell_{\mPa}^{\mPa}(j) - \ell_{\tilde\mPa}^{\mPa}(j) + I(\rho_{\tilde \mPa}^{\mPa}(j))\right|\\
      &\geq  c_\delta.
    \end{align*}
    The lower bound implies
    \begin{align*}
      c_\delta > |\mV{\mPa}{\mPa}{0}(x) - \mV{\mPa}{\tilde\mPa}{0}(x)| &\geq c_\delta
    \end{align*}
    yielding the contradiction.
  \end{proof}

\subsection*{Proof of Theorem~\ref{thm:total_number_of_unsafe_episodes}}
  Let $\NN_{\mathrm{unsafe}}$ be the set of episode indices such that
  $\mExpectationExp{s}{|\mV{\mPa}{\mPa}{0}(x) - \mV{\mPa}{\mPa_e}{0}(x)|} > c_\delta$
  for $e\in\NN_{\mathrm{unsafe}}$,
  i.e. potentially unsafe episodes for which there exists some $j$ s.t.
  $\mExpectationExp{E,x}{x_{\mPa_e}^\mPa(j)\notin\XX}$ out of a total of $N$ episodes.
  By summing up potentially unsafe episodes $e\in \NN_{\mathrm{unsafe}}$ we get
  \begin{align}\label{eq:thm_total_number_of_unsafe_episodes}
    \sum_{e\in\NN_{\mathrm{unsafe}}}\mExpectationExp{s}{|\mV{\mPa}{\mPa}{0}(x) - \mV{\mPa}{\tilde\mPa}{0}(x)|}
    \geq N_{\mathrm{unsafe}} c_\delta
  \end{align}
  by Lemma~\ref{lem:reward_constraint} with $N_{\mathrm{unsafe}}$ such that
  $|\NN_{\mathrm{unsafe}}| = N_{\mathrm{unsafe}}$. By definition,
  the cumulative regret provides a bound for the sum in~\eqref{eq:thm_total_number_of_unsafe_episodes} 
  and we therefore end up with $CR(N) \geq N_{\mathrm{unsafe}}c_\delta$, which proves
  the desired statement.\qed

\end{document}